\documentclass{article}

\usepackage{tikz,mdframed,natbib,amssymb,amsmath,amsthm}
\newtheorem{observation}{Observation}
\newtheorem{theorem}{Theorem}
\newtheorem{definition}{Definition}

\newtheorem{corollary}{Corollary}
\newtheorem{conjecture}{Conjecture}
\newtheorem{axiom}{Axiom}
\newtheorem{example}{Example}
\newtheorem{lemma}{Lemma}
\newenvironment{axiombis}[1]
  {\renewcommand{\theaxiom}{\ref{#1}$'$}%
   \addtocounter{axiom}{-1}%
   \begin{axiom}}
  {\end{axiom}}

\title {Being Central on the Cheap:
Stability in Heterogeneous Multiagent Centrality Games}
\author {Gabriel Istrate and Cosmin Bonchi\c{s}\footnote{West University of Timi\c{s}oara, Romania. email: gabrielistrate@acm.org}}

\begin{document}
\maketitle
\begin{abstract} We study strategic network formation games in which agents attempt to form (costly)  links in order to maximize their network centrality. Our model derives from Jackson and Wolinsky's \emph{symmetric connection model}, but allows for heterogeneity in agent utilities by replacing 
\emph{decay centrality} (implicit in the  J.-W. model) by a variety of classical centrality and game-theoretic measures of centrality. We are primarily interested in characterizing the asymptotically pairwise stable networks, i.e. those networks that are pairwise stable for all sufficiently small, positive edge costs. We uncover a rich typology of  stability: 
\begin{itemize} 
\item[-] we give an axiomatic approach to network centrality that allows us to predict the stable network for a rich set of combination of centrality utility functions, yielding stable networks with features reminiscent of structural properties such as "core periphery" and "rich club" networks. 
\item[-] We show that a simple variation on the model renders it universal, i.e. every network may be a stable network. 
\item[-] We also show that often we can infer a significant amount about agent utilities from the structure of stable networks. 
\end{itemize} 
\end{abstract} 

\section{Introduction}

Centrality in social networks is a topic that has seen an overwhelming amount of recent work at the intersection of Social Network Analysis \cite{wasserman1994social}, Physics of Complex Systems \cite{newman2018networks}, Economics \cite{jackson-networks}, Theoretical Computer Science and Artificial Intelligence \cite{easley2010networks}. Many of the leading models of network formation are stochastic. In reality, networks form and evolve as a consequence of \emph{agent incentives}. Among these incentives, \emph{centrality maximization} is certainly a pervasive one. To give just one all too familiar example: the increasingly competitive nature of the scientific enterprise, coupled perhaps with an ever more common\footnote{and, in our opinion, unfortunate} reliance on quantitative measures\footnote{of "centrality" in the citation network (!), such as the H-index} and rankings of individuals and publishing venues as proxies for research quality, has often resulted in a significant explosion in the number of submissions to venues perceived as "top ones". The decisive factor seems to be, of course, authors' perception that publication in such venues is a way to increase their papers' impact, which ultimately increases their own ''centrality", that they want maximized. 

Somewhat less clear are the \emph{strategic consequences} of agents' propensity for competing for central positions.  By this we mean understanding the manner  in which agents' preferences  for central network positions influences the formation and evolution of networks: in spite of a significant amount of work on the structure of networks (witnessed e.g. by concepts such as \emph{core-periphery structure} \cite{borgatti2000models}, \emph{the rich-club effect} \cite{zhou2004rich}, or \emph{small-world networks} \cite{watts:swn:book}), and on strategic network formation \cite{jackson-wolinsky,bala2000noncooperative,fabrikant2003network}, the impact of agents' incentives on  the network structure is not completely understood. On the other hand, the inverse problem, that of inferring node utilities from the structure of stable networks, has not seen a comparable amount of investigation. 

The goal of this paper is to contribute to understanding network formation from a game-theoretic perspective that assumes that \textbf{agents are willing to modify network structure (at a small cost per extra link) in order to improve their centrality.} We start from the realisation that the most well-known model of strategic network formation, the \emph{symmetric connection model} \cite{jackson-wolinsky,jackson-networks} can be seen as maximizing agents' \emph{decay centrality} \cite{hurajova20decay,tsakas2018decay}, subject to constant edge cost. In reality, centrality objectives  may differ from agent to agent. The extension we propose accommodates this heterogeneity via an axiomatic approach to network centrality \cite{boldi2014axioms}, presenting axioms that  isolate salient features of such measures, which influence the structure of emerging networks 
in predictible ways. 

A key design dimension for models of network formation is their handling of \emph{tie strength}: 
It is well-known  that weak ties  play an important role in social dynamics, being extremely useful for information dissemination \cite{granovetter1973strength}. Our model restricts agents tie manipulation attempts to \emph{weak ties only}, whose cost of establishing/maintainance can realistically be assumed to be a tiny positive constant. An explanation for such a restriction is not merely modeling convenience, but the existing tension between breaking strong links and centrality maximization: in real life strong links could be either exogeneously imposed (e.g. agents interacting as being part of the same organization, e.g. being coworkers) or have such a high intrinsic cost of breaking (e.g. family ties) that the objective of centrality maximization becomes unrealistic. 
 It is an interesting question to extend our framework to networks with strong ties as well.  

We follow the philosophy of \cite{jackson-wolinsky} (rather than that of  \cite{fabrikant2003network}) by assuming that our model is \emph{symmetric}: link establishment must benefit \emph{both} endpoints. The main questions motivating us are: 

 \begin{mdframed} 
 \begin{itemize} 
 \item[1.] Given a set of agents with diverse centrality objectives, can one predict the structure of stable nets ? 
 \item[2.] Can we obtain richer classes of stable networks than the ones for the classic symmetric connections model ? Are observed features of real-world graphs (e.g. rich-club, core-periphery) compatible with explanations based on centrality maximization ? 
 \item[3.] Can we infer (at least something about) agents'  centrality objectives from the (set of) stable networks ? 
 \end{itemize} 
 \end{mdframed} 
 
 Given the extreme potential heterogeneity of agent objectives, it could seem that no substantive positive answer could be given to the questions above. \textbf{A first contribution of our paper is to show that this is not the case}: we provide  well-behaved examples where the answers to the three previous questions are affirmative. Remarkably, stable networks in our models display (stylized versions of) features such as core-periphery and rich club effect, suggesting that suitable extensions of the J-W model may be able to reproduce such features. A second contribution is the identification of a \textbf{fundamental limit of prediction in centrality maximization models}, a plausible variation in the model specification that renders it \emph{universal}, in the sense that  \textbf{every network} can arise as an equilibrium. The property we highlight is natural, assuming that agents' utility  is subject to a threshold beyond which no extra gain in centrality increases their utility. 
 
 The following is an outline of our main results \textbf{(for all the missing proofs we refer to the Appendix):}  we first show (Theorem~\ref{stable-mon}) that for mixtures of agents satisfying monotonicity axioms, stable networks may contain an unique complex component consisting of a "core" (clique) and a "periphery" of nodes only connected to core nodes. On the other hand (Theorem~\ref{deghom}), for degree homophilic measures  stable networks \emph{display a rich club effect}: they consist of a sequence of  cliques of rapidly decreasing sizes.  
For betweenness centrality games we show (Theorem~\ref{bet-thm}) that connected stable networks have small diameter and are characterized by a property related to  \emph{neighborhood domination} in graphs. We provide some futher analytical results and conjectures (based on simulations) for some less well-behaved centralities. Finally, we show that adding thresholds to the model makes it universal (Theorem~\ref{thr-thm}). 
Some of these cases are complemented by results (Theorems~\ref{learn-mon},~\ref{learn-thr}) on learning agent types from the (set of) stable networks.  
\section{Preliminaries}

We assume familiarity with basics of graph theory, coalitional game theory (see \citet{chalkiadakis2011computational}), strategic network formation (e.g. \citet{jackson-networks}) and centrality measures in social networks \cite{koschutzki2005centrality,das2018study}. In particular, given network $g$ and vertex $i$ of $g$, we will denote by $deg(i)$ the degree of $i$ in $g$, by $N(i)$ the set of neighbors of $i$ in $g$, and by $\widehat{N(i)}$ the set $\{i\}\cup N(i).$ We also denote by $g+ij$ the network obtained by adding missing edge $ij$, by $d(i,j)$ the distance between $i$ and $j$, and by $Conn(i)$ the connected component of $i$ in $g$. We will write $g_1+g_2$ for the disjoint union of two networks $g_1,g_2$. An edge is called \emph{a bridge edge} if its removal disconnects the graph. The neighborhood domination relation is a classical concept in graph theory (e.g. Definition 1.16 in \citet{brandstadt1999graph}), formulated (as \emph{vicinal preorder}) in \citet{foldes1978dilworth}  as follows: 
Given vertices $x,y$, say that \emph{$y$ dominates $x$} (write $x\leq y$) iff 
$N(x)\subseteq N(y)\cup \{y\}$.

\begin{definition} Given node $i$ in network $g$, define: 
\begin{itemize} 
\item[-] The \emph{degree centrality of $i$} is  $
C_{deg}[i]=deg(i).$ 
\item[-] \emph{Linear centralities.} Given a set of nonnegative integral weights $w_{i,j}$ on $V$, define $C_{w}[i] = \sum_{ij\in E(g)} w_{i,j}$.  
\item[-] The \emph{closeness centrality of $i$} is  $
C_{close}[i]=\frac{1}{\sum\limits_{j\in Conn(i)} d(i,j)}.$ 
\item[-] The \emph{eccentricity centrality of $i$} is defined [\footnote{in many papers eccentricity is defined as $max(d(i,j):j\in Conn(g))$. Defining eccentricity centrality like we do has been done before, and has the advantage that bigger values correspond to "more central nodes" .}] as 
$C_{ecc}(i)=0$ if $i$ is isolated, $min(\frac{n-1}{d(i,j)}:j\in Conn(g))$ otherwise. 
\item[-]  The \emph{random walk closeness centrality of $i$} is defined \cite{white2003algorithms} as $C_{close}[i]=\frac{1}{\sum\limits_{j\in Conn(i)} ht[j,i]},$    
where $ht[j,i]$ is the expected time for a random walk started at $j$ to first hit $i.$ 

\item[-] The decay centrality of $i$
 is defined as $C_{dec}[i]=\sum_{j\in g} \beta^{d(i,j)}$, where $\beta$ is a fixed parameter, $0<\beta<1.$

\item[-] The \emph{harmonic centrality of $i$} is 
$C_{harm}[i]=\sum_{j\in g} \frac{1}{d[i,j]}.$
\item[-] The \emph{betwenness centrality of $i$} is 
$C_{between}[i]=\sum\limits_{y\neq i\neq z}
\frac{\sigma_{y,z}(i)}{\sigma_{yz}},$ 
that is the sum of percentages of shortest paths between arbitrary vertices $y,z$ that pass through $i$. 
\item[-]The \emph{random walk betwenness centrality} (a.k.a. \emph{current flow betweenness centrality}) \emph{of $i$} is defined \cite{newman2005measure} as 
$C_{RWB}[i]=\sum\limits_{j\neq i\neq k} r_{j,k},$  
where $r_{j,k}$ is the probability that a random walk  starting at node $j$ with absorbing node $k$ passes through node $i$. 
\item[-] The \emph{eigenvector centrality of $i$} is defined as 
$C_{eig}[i]=w[i]$, where $w$ is the eigenvector corresponding to the largest eigenvalue of the adj. matrix of $g.$ 
\item[-] The \emph{Katz centrality of $i$} is defined as 
$C_{Katz}(i)=\sum\limits_{k=1}^{\infty} \sum\limits_{j=1}^{n} \alpha^k (A^{k})_{ji}$. 
where $\alpha$ is a parameter, $0<\alpha<1.$
\item[-] Pagerank. See e.g. \citet{boldi2017rank} for formal definitions and some properties. 


\end{itemize} 
\end{definition} 


Our last centrality is defined using coalitional games: 

\begin{definition}The \emph{game-theoretic centrality of $i$} is defined as the Shapley value of node $i$ in the coalitional game $(N,v)$, where $v(S)=|S\cup N(S)|$. It has the formula \cite{michalak2013efficient} $C_{GT}[i]= \sum\limits_{j\in \widehat{N(i)}} \frac{1}{deg(j)+1}.$
\end{definition} 

\section{Model and Axiomatic Setting}

Our framework is specified as follows: 

\begin{definition}
The \emph{symmetric connection model with generalized centralities $(C_i)$ and edge cost $c$} is defined as follows: the utility of player $i$ on network $g$ is 
\begin{equation} 
u_{i}(g)=C_{i}(g) - c\cdot deg(i). 
\label{eq:utility}
\end{equation}   
We will occasionally avoid mentioning the family of centralities $(C_i)$ when it is clear from the context. When all $C_{i}$'s are equal to the decay centrality we recover the classical symmetric connection model \cite{jackson-wolinsky}. 
\end{definition} 

An \emph{edge flip of a pair of nodes $i,j$ of a network $g$} is the addition of $ij$ to $g$, if $ij\not \in g$, or its removal  from $g$, if $ij \in g$. The \emph{outcome of an edge flip} is the resulting network $h$. 

\begin{definition}
An edge flip is a \emph{weakly improving move for player $i$} if $u_{i}(h)\geq u_{i}(g)$, and  a \emph{strongly improving move} if $u_{i}(h)> u_{i}(g)$.  An edge flip is \emph{an improving move} iff: 
\begin{itemize} 
\item[-] it is an edge addition that is strongly improving for at least one endpoint and at least weakly improving for both, or 
\item[-] is an edge deletion, strongly improving for some endpoint. 
\end{itemize} 
\end{definition} 

The main model of stable network structure employed in the area of strategic network formation, defined in  \citet{jackson-wolinsky} is: 

\begin{definition} Network $g$ is called \emph{pairwise stable} if no edge flip is an improving move. 
\end{definition} 

In this paper we use a version of pairwise stability that is appropriate to our setting that assumes weak ties only, in which the edge cost is a tiny (but positive) value $\epsilon >0$. Therefore, the following variant of pairwise stability will be our main notion of interest: 

\begin{definition} Network $g$ is called \textbf{asymptotically pairwise stable} (APSN) if there exists $\epsilon_{0}>0$ such that $(\forall \epsilon)$,  $0<\epsilon<\epsilon_0$, $g$ is pairwise stable in the  model with edge cost $\epsilon$. 
\end{definition} 

APSN is a version of pairwise stability, so results on APSN relate to preexisting literature. For instance the original theorem of Jackson-Wolinsky implies the fact that for decay centrality games the unique family of APSN is formed by complete graphs $K_n$. 

\subsection{Axioms for Monotone Network Centralities} 

The first axiom is a simple one, and has been discussed before in the literature \cite{boldi2014axioms}. It formalizes the intuition that  adding edges always improves the centrality of adjacent nodes: 

\begin{axiom}A centrality measure $C$ is \emph{increasing} if whenever $ij\not \in g$, 
$C[i,g+ij]>C[i,g].$ 
\label{axiom-1}
\end{axiom} 

If $C$ is a centrality measure satisfying some axiom then $C^{\prime}=1/(1+C)$ (or even $C^{\prime}=1/C,$ when $C$ is strictly positive) satisfies a correspondingly modified "dual" axiom. For instance, here's the dual of Axiom~\ref{axiom-1}: 

\begin{axiombis}{axiom-1} 
A centrality measure $C$ is \emph{decreasing} if whenever $ij\not \in g$, $
C[i,g+ij]< C[i,g].$ 
\end{axiombis} 

Our next axiom encodes a different type of monotonicity: one in which the benefit of  extra links only incurs for agents already in the same connected component. It is a more precise version of an axiom due to \citet{boldi2017rank}:  

\begin{axiom} A centrality measure $C$ is \emph{componentwise} if it satisfies the following conditions: 
If $ij\not \in g$ and $i,j$ are in the same connected component then $C[i,g+ij]>C[i,g]$. If $ij\not \in g$ and $i,j$ are not in the same connected component in $g$ then   
$C[i,g+ij]\leq C[i,g]$.
\label{axiom-2}
\end{axiom} 

In the dual scenario agent only benefit when forming bridges between previously disconnected components: 

\begin{axiombis}{axiom-2}
A centrality measure $C$ is \emph{peripheral} if it satisfies the following conditions: 
If $ij\not \in g$ and $i,j$ are in the same connected component then $C[i,g+ij]\leq C[i,g]$.
If $ij\not \in g$ and $i,j$ are not in the same connected component in $g$ then   
$C[i,g+ij]> C[i,g]$.  
\end{axiombis}

The setting of Axiom 2 (and, consequently, $2^\prime$) is not vacuous,  as we have: 

\begin{theorem} 
\reversemarginpar
Closeness centrality and random walk closeness centrality satisfy Axiom 2. 
\label{cc-axiom} 
\end{theorem} 

The next axiom has a different flavor, and encode a scenario when agents benefit by connecting only when they were "of the same/different types". In our particular setting homophily is assessed with respect to agents' degree. 

\begin{axiom}
A centrality measure $C$ is \emph{degree homophilic} if the following is true: there exists a strictly increasing function $f$ such that for any network $g$ and edge $ij\not\in g$, adding edge $ij$ to $g$ is an improving move for $i$ iff $deg(j)\leq f(deg(i))$.
\label{axiom-ass} 
\end{axiom} 
An example of measure satisfying Axiom 3 is: 

\begin{theorem} 
Game-theoretic centrality satisfies Axiom~\ref{axiom-ass} with $f(0)=-1$, $f(n)=(n+1)(n+2)-3$, $n\geq 1$.   
\label{degree-seg} 
\end{theorem} 
\noindent Another axiom refers, among other things, to the role links among unrelated agents have on one agent's centrality: 

\begin{axiom}
Increasing centrality measure $C$ is called \emph{regular} if the following are true: 
\begin{itemize} 
\item[-] whenever $g$ is a network and $i$ an isolated node in $g$, $C[i,g]=0$. 
\item[-] whenever $i,j$ are nodes, $g$ is a network such that $ij\not\in g$ and $k\neq i,j$, we have 
$C[k,h+ij]\leq C[k,h]$ In other words the addition of an edge cannot increase the centrality of an unrelated node.  
\end{itemize} 
\label{axiom-ind} 
\end{axiom} 

\begin{example} 
Linear centralities are clearly regular.  
\end{example}

\begin{figure}
\begin{center}
\begin{tabular}{|c|c|c|}
\hline\hline
Centrality & Axm. & Reference \\
\hline\hline
Degree &  1,4 & trivial \\
\hline
Harmonic & 1 & \citet{boldi2014axioms}\\
\hline
Katz & 1 & trivial \\
\hline
Decay & 1 & trivial\\
\hline
Pagerank & 1 & \citet{chien2004link}  \\
 &  & \citet{boldi2017rank} \\
\hline
Closeness & 2 & Theorem~\ref{cc-axiom}
\\
\hline
r.w. Closeness & 2 & Theorem~\ref{cc-axiom}
\\
\hline
Game-theoretic & 3 & Theorem~\ref{degree-seg}\\
\hline
\end{tabular}
\end{center}
\caption{\label{table-centralities}Axioms for centrality measures.}   
\end{figure}

The table in Figure~\ref{table-centralities} displays a synthetic view of the axioms various centrality measures satisfy. We will say that a game satisfies a certain axiom if all the players satisfy it.  

Finally, many of our results (even theoretical ones) arise from implementing our model in Python using the \emph{networkx} package \cite{hagberg2008exploring} and performing computational experiments. We will provide a link to public versions of our programs in the non-anonymized version.

\section{APSN for Mixtures of Monotone Centralities}


In this section we consider mixtures of agents satisfying axioms $1,1^{\prime}$, and $2,2^{\prime}$. For latter types improving moves are clear by definition. The following result, whose proof is trivial,  characterizes them for the former types as well: 

\begin{lemma} Consider an agent $i$ whose utility function has the form in equation~(\ref{eq:utility}). The following is true: If centrality $C$ satisfies Axiom 1 then adding any missing edge is an improving move for $i$, while removing an edge is not improving, both statements being true for small enough cost $c$. Dually, if centrality $C$ satisfies Axiom $1^{\prime}$ then removing any edge is an improving move for $i$, while adding an edge is not improving, both statements being true for small enough cost $c$.  
 \label{pred-1}
\end{lemma}

The next result characterizes stable networks as those whose connected components are (perhaps single node) complete graphs plus, maybe, one \emph{complex component}, which is not a complete graph. This complex component displays an extreme form of \emph{core-perifery structure}: it consists of a \emph{core},  a clique of nodes of type 1 and 2, and a \emph{periphery} consisting of nodes of type $2^{\prime}$ attached to core nodes only:  

\begin{theorem}
The APSN in the centrality model with agents satisfying one of Axioms 1,$1^{\prime}$, 2,$2^{\prime}$ are precisely those networks $g$ satisfying the following rules: 
\begin{itemize}
\item[-] All agents of type $1^{\prime}$ are isolated.
\item[-] There is a single connected component that contains all agents of type 1. Agents of type 1 and 2 in this component form a clique ("the core"). All agents of type $2^{\prime}$ belong to this component, and are pendant vertices attached to vertices of type 1 of the clique ("the periphery"). 

\item[-] All other connected components are complete graphs (including isolated nodes) containing agents of type 2 only. 
\end{itemize} 
\label{stable-mon}
\end{theorem} 
\begin{proof} Consider an APSN. 
All agents of type 1 must be connected in a clique, since joining them is an improving move for all. For agents of type 2 belonging to this component it is beneficial to connect to all nodes of type 1 (and among themselves), hence they are also part of the clique core.  Agents of type $2^{\prime}$ want to stay connected to the component, but only minimally: once they are connected to a node in the component, adding any extra edge is not improving for them. Their contacts must be of type 1: if they were of type 2 they'd benefit from severing the connection. 

All other components consist of agents of type 2 only, for which it is beneficial to fully connect. 
\end{proof}

\begin{figure}[h]
\begin{center}
\scalebox{0.450}{

\begin{tikzpicture}[shorten >=1pt, auto, node distance=3cm, ultra thick]
 	\tikzstyle{type1} = [circle, draw=red, fill=white!]    
    \tikzstyle{type11} = [circle, draw=blue, double, fill=white!]  
    \tikzstyle{type2} = [rectangle, draw=red, fill=white!]  
    \tikzstyle{type22} = [rectangle, draw=blue, double, fill=white!]  
    \tikzstyle{edge_style} = [draw=black, line width=2, ultra thick]
    
	\node[type2] (v1) at (-2,3) {1};    
    \node[type1][pin=50:{$x$}] (v2) at (0.5,0) {2};
    \node[type1] (v3) at (-4,1) {3};

    \node[type2] (v4) at (-2,-1) {4};
	\node[type1](v5) at (1,3) {5};

    \node[type2] (v9) at (-5,-1) {9};
    \node[type11] (v8) at (-1,5) {8};
    \node[type22] (v15) at (-5,3) {15};
    
    \node[type22] (v13) at (1, 5) {13};
     \node[type22] (v14) at (3, 3) {14};
    \node[type22] (v6) at (2, 4) {6};
    \node[type22] (v7) at (-6,1) {7};
    
    \node[type2] (v10) at (3,1) {10};
	\node[type2][pin=90:{$y$}] (v11) at (5,3) {11};
    \node[type2] (v12) at (5,-1) {12};

	\foreach \x [count=\xi from 2] in {1,...,4}{%
    	\foreach \y in {\xi,...,5}{
    		\draw[edge_style]  (v\x) edge (v\y);
        }
    }
    \foreach \x [count=\xi from 11] in {10,11}{%
    	\foreach \y in {\xi,...,12}{
    		\draw[edge_style]  (v\x) edge (v\y);
        }
    }
   	
   	\draw[edge_style]  (v5) edge (v6);
    \draw[edge_style]  (v3) edge (v7);
    \draw[edge_style] (v5) edge (v13); 
     \draw[edge_style] (v5) edge (v14);
     \draw[edge_style] (v3) edge (v15);
  
    \matrix [draw,below left] at (current bounding box.north west) {
  		\node [type1,label=right:$type~1$] {}; \\
  		\node [type11,label=right:$type~1'$] {}; \\
  		\node [type2,label=right:$type~2$] {}; \\
  		\node [type22,label=right:$type~2'$] {}; \\
	};
\end{tikzpicture}}
\end{center}
\caption{\label{fig:T3}APSN for mixtures of  monotone centralities.}
\end{figure}

As a corollary of the previous result, we can infer quite a lot about agent types from the structure of stable networks: 

\begin{theorem} 
Consider, in the setting of Theorem~\ref{stable-mon}, an asymptotically pairwise stable network $g$. Then: 
\begin{itemize} 
\item[a).] Isolated nodes are either of type $1^{\prime}$ or $2$. 
\item[b).] An agent $y$ in complete graph components of size $\geq 2$ may be of type $1$ or $2$. It is guaranteed to be of type 2 when there exists a complex component in $g$ which doesn't contain $y$, or when some agent $x$ in a different component than $y$ is known to be of type 1 (Figure~\ref{fig:T3}). 
\item[c).] In the complex component pendant vertices are of type $2^{\prime}$ (Figure~\ref{fig:T3}) and their unique neighbors are of type 1.  All other nodes  are of types 1 or 2. 
\item[d).] One cannot distinguish between centrality measures satisfying the same Axiom, nor between the multiple possible types of nodes in the same listing (a)-(c).
\end{itemize} 
\label{learn-mon}
\end{theorem} 

\section{Degree-homophily yields rich-club APSN} 

Next we study centrality games for degree-homophilic centrality measures. The following result shows that  in this case APSN have a "rich club", hierarchical structure: 

\begin{theorem} 
Let $h$ be an APSN for the centrality game with degree homophilic centralities with function $f(\cdot)$ which satisfies $f(0)=-1$ and $f(x)\geq x$ for every $x\geq 1$. 

 Let $m$ be the maximum degree of a node in $h$. Let $n_{1}^{*}=min\{k:f(k)\geq m\}$ and, for $i\geq 2$, $n_{i}^{*}=min\{r:f(r)\geq n_{i-1}^{*}\}$. Clearly $n_{1}^{*}\geq n_{2}^{*}\geq \ldots$ (and one can assume w.l.o.g., by removing multiple copies of the same value, that $n_{1}^{*}>  \ldots > n^{*}_{r}=1$ for some $r\geq 1$)
\begin{itemize} 
\item[a).] If $deg(i),deg(j)\geq n_{1}^{*}$ then $ij\in E(h)$. 
\item[b).] If $deg(i),deg(j)\in [n_{k}^{*},n_{k-1}^{*}]$ for some $k\geq 2$ then $ij \in E(h)$ ("alike nodes connect to each other")
\item[c).] If $k\geq 2$, $deg(i)\leq  n_{k}^{*}$, $deg(j)> n_{k-1}^{*}$ then $ij\not \in E(h)$. 
\end{itemize} 
\label{deghom} 
\end{theorem} 
\begin{corollary} 
Let $a_{1}>a_{2}> \ldots > a_{p}> a_{p+1}=1$ be a sequence of integers such that $a_{i}-1 > f(a_{i+1}-1)$ for all $i=1,\ldots, p$. Then all graphs of type 
$K_{a_{1}}+K_{a_{2}}+\ldots + K_{a_{p}}+rK_{0}$  ("stratified clique graphs") are APSN for degree homophilic centrality games with function $f$ and conversely, all APSN are unions of cliques with this structure. 
\label{rich-thm} 
\end{corollary} 
\begin{proof} 
We use the definition of degree homophily: 
\begin{itemize} 
\item[a).] Since $deg(i)\geq n_{1}^{*}$ and $f$ is monotonic, $f(deg(i))\geq f(n_1^{*})\geq m\geq deg(j)$, and similarly $f(deg(j))\geq deg(i)$. If $i,j$ were not connected, then adding $ij$ would be an improving move for both of them. 
\item[b).] Similar to  (a): Since $deg(i)\geq n_{k}^{*}$ and $f$ is monotonic, $f(deg(i))\geq f(n^{*}_{k}) = n^{*}_{k-1}\geq  deg(j)$, so $f(deg(i))\geq deg(j)$, and similarly $f(deg(i))\geq deg(j)$. If $i,j$ were not connected, then adding $ij$ would be an improving move. 

\item[c).] 

We have $deg(i)\leq n_{k}^{*}$ so $f(deg(i)-1)\leq f(n_{k}^{*}-1)<n_{k-1}^{*}\leq deg(j)-1$.  In conclusion, $deg(j)-1> f(deg(i)-1)$, which means that removing edge $ij$ is an improving move for $j$, since in the graph $h=g-ij$ adding edge $ij$ is not an improving move for $j$.  
\end{itemize} 

For a proof of the corollary, see the Technical Appendix. 
\end{proof}

\section{Domination and APSN for Betwenness Games} 

In this section we characterize APSN for centrality games for betweenness centrality, a measure that satisfies none of the previous axioms. First, simple computations provide examples of APSN with components that are not complete graphs: networks $C_{4}+nK_{1}$, $n\geq 0$.  We will show that the domination relation plays a decisive role in their characterization of APSN. To accomplish this, we first prove: 

\begin{lemma} The following statements are true: 
\begin{itemize} 
\item[-] Adding any bridge edge $ij$ weakly increases $i$'s betweenness centrality, strictly  unless $i$ was an isolated node. Consequently adding a bridge edge is improving for $i$, unless $i$ was isolated. Conversely, a disconnecting edge removal would be improving for $i$ iff $i$ was a pendant node. 
\item[-] Adding any non-bridge edge $ij$  weakly increases $i$'s betweenness centrality. 
\end{itemize} 
\label{bet} 
\end{lemma} 

We now prove the following result, which gives an unexpected (and fairly elegant) algorithmic characterization of APSN for betweenness games using the domination relation: 

\begin{theorem} 
Graphs $g$ that are APSN for betweenness centrality games consist of isolated vertices plus at most one connected component $C$ with at least two vertices which satisfies the following condition: $deg(l)\geq 2$ for every $l\in C$, $diam(C)=2$ and for every $i\neq j\in C$, $ij\in E(g)$ \textbf{if and only if} sets $N(i)\setminus \{j\}$ and $N(j)\setminus \{i\}$ are incomparable, i.e. if \textbf{none of $i,j$ dominates the other.}
\label{bet-thm} 
\end{theorem} 
\begin{proof} 

It is easy to see that the networks that satisfy the condition of Theorem~\ref{bet-thm} are APSN: first, by Lemma~\ref{bet} isolated vertices have no incentive to connect to anyone else, as their utility would decrease, and can be subsequently ignored.  Consider, two vertices $i,j$ in a larger component $C$. 
If $ij\in E(g)$ then, by the condition of the theorem, there exist vertices $k\in N(i)\setminus N(j)$ and $l\in N(j)\setminus N(i)$. Since $d(k,j)\leq 2$ it follows that $[k;i;j]$ is a shortest path between $k$ and $j$ that would disappear if we dropped edge $ij$, decreasing the betweenness centrality of $i$ and ultimately its utility. Similar arguments hold for $j$. So if we dropped edge $ij$ the utility of both $i,j$ would decrease, hence dropping $ij$ it is \textbf{not} an improving move.  On the other hand if $ij\not \in E(g)$ then $N(i)\setminus \{j\}$ and $N(j)\setminus \{i\}$ are comparable. Assume w.l.o.g. that $N(i)\setminus \{j\}\subseteq N(j)\setminus \{i\}$. Then every shortest path  $[s;i;t]$ stays a shortest path when we add edge $ij$: This is clear when $s,t\neq j$, while cases $s=j$ and $t=j$ cannot happen, since $ij\not\in E(g)$. Adding edge $ij$ creates no new shortest paths: Indeed, assume $[s;i;j]$ were a newly created shortest path through $i$. Since $N(i)\setminus \{j\}\subseteq N(j)\setminus \{i\}$, $sj\in E(g)$, a contradiction. Since the set of shortest paths through $i$ stays the same, adding $ij$ is not an improving move. 

Let us now prove the converse direction, that APSN satisfy the conditions in the theorem. A first thing to prove is that any APSN has at most one component with $\geq 2$ vertices. Indeed, if there existed two such connected components then, by Lemma~\ref{bet}, joining them by an edge $ij$ would be an improving move for both $i,j$.  Second, we claim that this nontrivial component has diameter 2: indeed, it cannot have diameter 1, as complete graphs are not APSN. Assume it contained a shortest path of length 3 $[p;q;r;s]$. Then $p,s$ would increase their utility by connecting since, e.g., now there is a shortest path from $q$ to $s$  through $p$. Third, this component has no pendant nodes: any such node would have, by Lemma~\ref{bet}, an incentive to disconnect. 

To prove the domination condition, consider a connected APSN $g$ and a pair $ij\not \in g$, and assume w.l.o.g. that adding edge $ij$ decreases the utility of $i$ for small $\epsilon >0$, so that the move is not improving. Hence 
adding $ij$ does not increase the betweenness of $i$. Paths that contributed positively to the betweenness of $i$ \textbf{before adding $ij$} are between unconnected nodes $k_1,k_2\in N(i)$, so that a shortest path between $k_1,k_2$ goes through $i$. Then adding edge $ij$ does not change the ratio corresponding to $k_1,k_2$ in the betweenness of $i$. 
Consider now the shortest paths between $k_1\in N(i)\setminus \{j\}$ and $j$. Since the betweenness of $i$ does not increase as a result of adding edge $ij$, $k_1$ must be connected to $j$. Hence $j$ dominates $i$. 

Consider now an edge $ij\in g$. Since the removal of  $ij$ is not improving for either $i$ or $j$, there exists a shortest path between some vertices $s_1\neq i\neq t_1$ that uses edge $ij$. 
As $diam(g)=2$, one of $s_1,t_1$ (say $t_1$) must be $j$, hence $s=s_1\in N(i)\setminus N(j)$. Similarly, there must be a vertex $t=t_{2}$ in  $N(j)\setminus N(i)$. Hence none of $i,j$ dominates the other.

\end{proof}

\begin{figure}
\begin{center}
\scalebox{0.350}{

\begin{tikzpicture}[shorten >=1pt, auto, node distance=3cm, ultra thick]
 	\tikzstyle{type1} = [circle, draw=red, fill=white!]    
    \tikzstyle{type11} = [circle, draw=red, double, fill=white!]  
    \tikzstyle{type2} = [rectangle, draw=red, fill=white!]  
    \tikzstyle{type22} = [rectangle, draw=red, double, fill=white!]  
    \tikzstyle{edge_style} = [draw=black, line width=2, ultra thick]
    
	\node[type1] (v1) at (5,0) {1};    
    \node[type1] (v2) at (-2, 4) {2};
    \node[type1] (v3) at (-5,-2.5) {3};
	\node[type1] (v4) at (1,4) {4};    
    \node[type1] (v5) at (-3,-4) {5};
    \node[type1] (v6) at (-4,2) {6};
	\node[type1] (v7) at (1,-4.5) {7};    
    \node[type1] (v8) at (-5.5,0.5) {8};
    \node[type1] (v9) at (3.5,-2.5) {9};
	\node[type1] (v10) at (4, 2) {10};
	\foreach \x in {2, 3, 4, 5, 6, 8, 10}{%
    		\draw[edge_style]  (v1) edge (v\x);
    };
    \foreach \x in {3, 5, 7, 9}{%
    		\draw[edge_style]  (v2) edge (v\x);
    };
  	\foreach \x in {4, 6, 7, 8, 9}{%
    		\draw[edge_style]  (v3) edge (v\x);
    };
    \foreach \x in {5, 7, 9}{%
    		\draw[edge_style]  (v4) edge (v\x);
    };
    \foreach \x in {6, 7, 8, 9}{%
    		\draw[edge_style]  (v5) edge (v\x);
    };
  	\foreach \x in {7, 9}{%
    		\draw[edge_style]  (v6) edge (v\x);
    };
    \foreach \x in {8, 10}{%
    		\draw[edge_style]  (v7) edge (v\x);
    };
    \foreach \x in {9}{%
    		\draw[edge_style]  (v8) edge (v\x);
    };
  	\foreach \x in {10}{%
    		\draw[edge_style]  (v9) edge (v\x);
    };
    
\end{tikzpicture}}
\end{center}
\caption{Non-bipartite APSN for betweenness games.} 
\label{cexp} 
\end{figure}
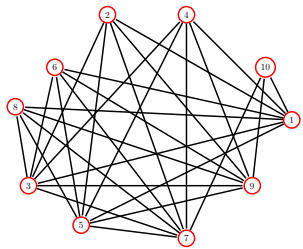

\begin{observation} 
Complete bipartite graphs $K_{a,b}$, $a,b\geq 2$ satisfy the conditions of the theorem, hence they are APSN. 
One could conjecture that these are \textbf{all} connected APSN with at least $2$ vertices, but this is \textbf{not} true: for a counterexample, found using computer simulations,  see the graph $g$ in Figure~\ref{cexp}. Graph $g$ is not bipartite as it has, e.g. triangle 4,5,9. 
\end{observation}


\section{Other Centralities: Results and Conjectures}

In this section we study centrality games for some measures that appear to satisfy none of Axioms 1,2,3: eccentricity centrality, random walk betweenness and eigenvector centrality.  

First, we show that eccentricity centrality is very close to obeying Axiom 2: 

\begin{lemma} 
Let $g$ be a network, $i$ a node in $g$ and $j$ another node such that $ij\not \in g$. The following are true: 
\begin{itemize} 
\item[-] If $j\not \in Conn(i)$ then $EC(i,g+ij)\leq EC(i,g)$. 
\item[-] If $j\in Conn(i)$ then $EC(i,g+ij)\geq EC(i,g)$. The inequality is strict iff $j$ is on all shortest paths in $g$ to all vertices $k$ that are farthest from $i$. 
\end{itemize} 
\label{lemma-ecc}
\end{lemma} 

In spite of this result, the structure of APSN for eccentricity centrality games is quite different from the one for centrality games with measures satisfying Axiom~\ref{axiom-2}: 

\begin{theorem}
All vertices in connected components of size at least three  of an APSN have degree at least two. 
On the other hand all  connected, eccentricity-one graphs with min. degree 2 and at least two nodes with degree at most $n-2$ are APSN. There exist  (Fig.~\ref{ecc}) APSN with eccentricity two. 
\label{ecc-thm}
\end{theorem} 

\begin{figure}[h]
\begin{center}
\scalebox{0.350}{

\begin{tikzpicture}[shorten >=1pt, auto, node distance=3cm, ultra thick]
 	\tikzstyle{type1} = [circle, draw=red, fill=white!]    
    \tikzstyle{type11} = [circle, draw=red, double, fill=white!]  
    \tikzstyle{type2} = [rectangle, draw=red, fill=white!]  
    \tikzstyle{type22} = [rectangle, draw=red, double, fill=white!]  
    \tikzstyle{edge_style} = [draw=black, line width=2, ultra thick]
    
	\node[type1] (v4) at (-4,-3) {4};    
    \node[type1] (v5) at (0.5,1) {5};
    \node[type1] (v6) at (4,1.5) {6};
	\node[type1] (v7) at (3,-2) {7};    
    \node[type1] (v8) at (-0.5,-2.5) {8};
    \node[type1] (v9) at (-3.5,2) {9};

    \draw[edge_style]  (v4) edge (v5);
    \draw[edge_style]  (v4) edge (v8);
    \draw[edge_style]  (v8) edge (v7);
    \draw[edge_style]  (v7) edge (v6);
    \draw[edge_style]  (v6) edge (v5);
    \draw[edge_style]  (v5) edge (v9);
    \draw[edge_style]  (v9) edge (v8);    
    
\end{tikzpicture}}
\end{center}
\caption{APSN for eccentricity centrality games.} 
\label{ecc} 
\end{figure}
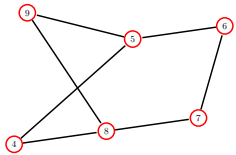 


We weren't able to obtain a full characterization of APSN in this case, or analytical results for random walk betweenness and eigenvector centrality. However, computer simulations suggest that the following statements are true: 

\begin{conjecture}
For random walk betweenness centrality the only APSN are the empty graph $\emptyset_{n}$ and the complete $K_{n}$.  
\end{conjecture} 

As for eigenvector centrality, although it seems not to have any monotonicity properties, experimental evidence is consistent with the following conjecture, that seems to situate this measure together with the monotonic ones: 

\begin{conjecture} 
The complete graphs $K_n$ are the only asymptotically pairwise stable networks for eigenvector centrality. 
\end{conjecture}

\section{Models with Truncated Centralities: Universality, Existence, Inference}

So far we have assumed that agents are unrestricted optimizers: they increase their centrality, subject to maintaining reasonable costs for their direct contacts. In reality, agents' utilities may be subject to \emph{diminishing returns}: the marginal benefit  from increased centrality diminishes (or even plateaus) beyond a point.  A simple way of incorporating this observation into the model is via the following: 

\begin{definition} Given centrality measure $C$ and threshold $\theta$, the \emph{$\theta$-truncation of $C$} is the centrality $C_{\theta}$ defined by 
\[
C_{\theta}[g]=\left\{\begin{array}{cc}
C[g], & \mbox{ if }C[g]<\theta \\
\theta, & \mbox{ otherwise.}
\end{array} 
\right. 
\]
Note that for $\theta=\infty$ (or just a large integer) we get the original utility. So truncations really extend our previous framework. 
\end{definition} 

This simple variation has a profound effect on  APSN: the following result first shows that truncation makes the model \emph{universal}. Then we give a condition for recognizing APSN in games with increasing measures. 

\begin{theorem} 
The following are true: 
\begin{itemize} 
\item[a).] For \textbf{every network $g$} and family of increasing centralities $C_{i}$ there exist thresholds $(\theta_{i})$ s.t. $g$ is an APSN for the truncated centrality game with thresholds $\theta_{i}$. 
\item[b).] On the other hand, if agents' original centrality measures are all increasing, then for all families $(\theta_{i})$ of thresholds, the APSN for the truncated centrality games with thresholds $\theta_{i}$, \textbf{if they exist}, can be characterized as the graphs with "Pareto optimal centralities", i.e. graphs $h$ satisfying the following properties: 
\begin{itemize} 
\item[-] 
for every  $ij\not \in E(h)$, $C_i[h]\geq \theta_{i}$ or $C_j[h]\geq \theta_{j}$, and 
\item[-] for every edge $ij\in E(h)$,  removing $ij$ from $h$ would yield a network $l$ with $C_i[l]<\theta_{i}$ and $C_j[l]<\theta_{j}$.
\end{itemize}  
\item[c).] \textbf{APSN exist} in all truncated centrality games with linear centralities. 
\item[d).] Consider a truncated centrality game with regular centralities $(C_{i})$ and nonnegative thresholds $(\theta_i)$. Let $M_i = max\{C_{i}[g+ij]: C_{i}[g]<\theta_{i}\}$. Assume that  whenever $g$ is a network, $i,j$ are nodes such that $ij\not\in g$ we have $C_i[g]<\theta_i \Leftrightarrow C_i[g+ij]\leq  M_i$ (note that the left-to-right implication is trivial). 
Then \textbf{APSN exist for such games.} 


\end{itemize} 
\label{thr-thm} 
\end{theorem} 

In spite of the previous result, we can still talk about learning agent utility functions. However, now we will not want to learn agent centralities (which we will, in fact, assume known), but   \textbf{agent thresholds.} The learning model we will assume is a type of \emph{oracle learning} \cite{ang:j:learning}.  In our case oracle  queries are pairs $(g,i)$ consisting of a network $g$  and an agent index $i$. Given query $(g,i)$ the oracle will either reply with an APSN $h$ such that 
$C_{i}[h]>C_{i}[g]$, or with "NONE", in the case such an APSN $h$ does not exist. 

Thresholds may fail to be identifiable in this model. One reason is \emph{the coarse resolution of our model}: for instance, any values between two consecutive integers (e.g. 2.3 and 2.7) are completely equivalent as thresholds for degree centrality, since degrees in graphs are integral, and jumps in centrality (as a result of an edge flip) have a magnitude at least one. Also, the fact that (consistent with the model in our Theorem~\ref{learn-mon}) we only get APSN as oracle answers may further impede the precision identification of $\theta_i$. One such "bad" case is when for all APSN $g$, $C_i[g]<\theta_{i}$: in this case all estimates provided by the oracle on the value of the threshold are too low. Nevertheless we can prove the following result that, in a sense, provides the best we can hope for: a numerical interval (arising from a single edge flip) encompassing the threshold value: 

\begin{theorem} Given an agent $i$, assume that there exists an APSN $h$ with $\theta_{i}\leq C_{i}[h]$. Then 
there exists an algorithm that uses oracle queries and outputs an APSN $g$ and an edge $ij\in E(g)$ such that $\theta_i\in [C_{i}[g-ij],C_{i}[g]]$. If $C_i$ is a linear centrality than this algorithm runs in polynomial time. 
\label{learn-thr}
\end{theorem} 
\section{Related Literature} 

The area of network games is quite large, and a comprehensive survey is impossible. We list here two such overviews: the first one, most relevant to our interests, is due to  \citet{jackson2005survey}. Another one with an algorithmic bent is due to \citet{tardos2007network}. The model that we are concerned with is a variant of the symmetric connection model \cite{jackson-wolinsky} (see also \citet{jackson-networks}). Some notable subsequent work includes \cite{dutta1997stable,jackson2002evolution,gilles2000original}. For an alternative model see  \citet{bala2000noncooperative}. 

Our work owes much to the axiomatic approach to network centralities. There has been significant work in this area both in the Theoretical Computer Science and A.I. communities \cite{boldi2014axioms,boldi2017rank,skibski2016attachment,skibski2017axiomatic}; \\\cite{bandyopadhyay2017generic,skibski2018axioms,wkas2018axiomatization}. 

A lot of related work exists in the theoretical computer science literature: for example, \citet{fabrikant2003network} (and a whole host of papers that extend this model) consider a game where edges are added at the incentive of \textit{one} of the endpoints. In contrast, our model takes a pairwise perspective, similar to \citet{jackson-wolinsky}. Other relevant work in this direction includes  
 \cite{hopcroft2008network}, who discuss an oriented model in which  nodes have control over outgoing edges. There is no cost for changing their links, and their purpose is to increase their Pagerank. They show that 
the Nash equilibria in this game have a fairly sophisticated structure (see also \citet{chen2009alpha}). Undirected versions of this game have been studied  as well \citep{avis2014reputation}. 
On the other hand \citet{avin2018preferential} prove that preferential attachment models can be seen as Nash equilibria of some network games. 

Other related work comes from the sociology literature \cite{hummon2000utility,buskens2008dynamics,narayanam2011topologies}. For instance, in the Buskens and Van De Rijt model every node strives to fill "structural holes" (including lack of connnectedness) between nodes. This is somewhat analogous to maximizing betweenness, but the precise model (and the results) are different. 

Our model allows heterogeneity in agents' utilities, corresponding to distinct measures of centrality. Heterogeneous network formation models have been studied previously, e.g. \citet{galeotti2006network}. 

Finally, several papers (e.g. \citet{crescenzi2016greedily,bergamini2018improving}) have treated the problem of improving the centrality of a node by adding or removing links. Our work is different in several respects: first of all, in our setting \emph{all} agents aim to improve their respective centralities. Second, in our model maintaining a link has a (small) cost. 

\section{Conclusions and Possible Extensions} 

We have shown that our models can accommodate a wide range of agent centrality objectives. Still, we do not see our results as adequate enough yet for the analysis of real-life networks. They have, instead, more of a proof-of-concept nature, showing that realistic network topologies could arise from strategic models of centrality maximisation, and could conceivably be made more realistic in many ways. Some variations (we believe) worth investigating are listed below:  

\noindent \textbf{Probabilistic edge addition/removal:} In real life an edge may only form with some probability even though both agents would benefit from it. Studying such a variation would be interesting, especially as it could produce networks with core structures that are dense but not quite complete. 

\noindent \textbf{Strong and weak links, forced links, affiliation models: }  In the version that we have discussed all the links are weak links. A natural extension allows for both strong and weak links. This would entail using two types of costs:  fixed, constant costs for the strong links, small ("$\epsilon$") costs for weak links.  A second, orthogonal,  distinction that could be useful is that of forced versus free links. 
Forced links may be a consequence of \emph{affiliation}: people meet as the result of joining the same clubs.  A possibly relevant model  is the \emph{social effort model} of \cite{borgs2010hitchhiker}. Another one is the \emph{social clubs} model of \cite{fershtman2018social}. For centrality in affiliation networks see \citet{faust1997centrality}. 

\noindent \textbf{Manipulating link strength: } In such a scenario presumably the agent could manipulate \emph{the strength} of the link. Whether to completely sever a link is a different decision. Further extensions could incorporate such decisions. 

\noindent \textbf{Tagged networks: } One possibility is that agents have a \emph{tag} and care about the tags of their neighbors/components. The most natural model is that they care about their neighbors, like in Schelling's segregation model. 

\noindent \textbf{Spatial agents: } Agents interaction may result from placement in space. A standard reference for the spatial version of the connection model is \\ \cite{gilles2000original}. 

\noindent \textbf{Multilayer networks: }  Sometimes (e.g. \citet{dickison2016multilayer}) link formation may encompass multiple, correlated, link types. E.g. two coworkers may end up being friends as well. It would be interesting to formulate multilayer models of network formation \'a la Jackson-Wolinsky. 

\noindent \textbf{Overlapping community structure: }  networks may have overlapping communities. For centrality in such models see \citet{szczepanski2014centrality}; \\ \citet{tarkowski2016closeness,gupta2016centrality,ghalmane2019centrality}. 

\noindent \textbf{Dynamic models: } Finally, our concepts of network stability are steady-state concepts. It would be interesting to study the emerging networks in \emph{dynamic} models of network formation with a similar philosophy.

Finally, our work leaves a large number of issues open: we list, for instance, only one of them: can one define (truncated) centrality games for which APSN fail to exist ? 




\newpage


\section*{Appendix: Deferred Proofs} 

\subsection{Proof of Theorem 1}

For closeness centrality games the formula for agents' utility is: 

\[
u_{i}(g)= \frac{1}{\sum\limits_{j\in Conn(i)} d(i,j)} - c \cdot deg(i)
\]

So, if $j\in Conn(i)$, $ij\not \in g$ and $h=g+ij$ we have 
\begin{align*}
& u_{i}(h)-u_{i}(g)=    \frac{1}{1+\sum\limits_{k\in Conn(i),k\neq j} d_{h}(i,k)}- 
- \frac{1}{d_{g}(i,j)+\sum\limits_{k\in Conn(i),k\neq j} d_{g}(i,k)}>0. 
\end{align*} 
The inequality holds since $d_{g}(i,j)>1$ and $d_{h}(i,k)\leq d_{g}(i,k)$ for all $k$. 

If  $j\not \in Conn(i),  u_{i}(g+ij)-u_{i}(g)=$ 
\begin{align*} 
&  \frac{1}{\sum\limits_{k\in Conn(i)} d(i,k)+ \sum\limits_{k\in Conn(j)} (d(j,k)+1)} - \frac{1}{\sum\limits_{k\in Conn(i)} d(i,k)}<0. 
\end{align*}
so closeness centrality satisfies Axiom 2. 

For random walk closeness centrality a fairly similar proof works: First consider the case $j\not \in Conn(i).$ Then the expected hitting times of nodes $H_{ki}$ of nodes $k\in Conn(i)$ do not change as a result of adding edge $ij$. Indeed, a random walk started in $Conn(i)$ cannot reach $j$ without reaching $i$. 

On the other hand in the formula of random walk closeness centrality we get terms corresponding to the new $k\in Conn(j)$ that can now reach $i$ as a result of adding edge $ij$. 

Consider now the case when $j\in Conn(i).$ Then it is true that $H^{\prime}_{ki}\leq H_{ki}$ for every $k\in Con(i)$ \cite{lovasz1993random}. Therefore $C_{i}(g+ij)\geq C_{i}(g)$. In fact the inequality is strict. We did not find a reference to quote, but the derivation is very easy: from elementary considerations. Let, indeed $N(i)=\{n_{1},\ldots n_{p}\}$. We have: 
\begin{align*}
& Pr[T_{k,i}>n]=\sum\limits_{v\in V}\sum_{l,s_{1},\ldots, s_{p}\geq 0}^{l+s_{1}+\ldots + s_{p}\leq n} \Pr_{k}[X_{n+1,v}^{(l,s)}]\cdot  \\
& \cdot \big(1-\frac{1}{deg(j)}\big)^{l} \prod_{\alpha=1}^{p}[1-\frac{1}{deg(r_{\alpha})}]^{s_{\alpha}}
\end{align*} 
where in the right-hand side $X_{n+1,v}^{(l,s)}$ is the event that at time $n+1$ the random walk (run \textbf{in graph $h=g\setminus \{i\}$}) is in state $v$ after having been before $l$ times in  $j$ and $s_{1},s_{2},\ldots s_{p}$ times in $n_{1},v_{2},\ldots, n_{p}$. 
 
When adding an edge from $j$ to $i$ we get a corresponding formula: 
\begin{align*} 
& Pr[T_{k,i}^{\prime} >n]= \sum_{v\in V}\sum_{l,s_{1},\ldots, s_{p}\geq 0}^{l+s_{1}+\ldots + s_{p}\leq n} \Pr_{k}[X_{n+1,v}^{(l,s)}]\cdot \\
& \cdot \big(1-\frac{1}{deg(j)+1}\big)^{l} \prod_{\alpha=1}^{p}(1-\frac{1}{deg(r_{\alpha})})^{s_{\alpha}}
\end{align*} 
The difference between the two formulas lies in the factor $(1-\frac{1}{deg(j)+1})^{l}$ in the second product. This is because in $g+ij$ we also have to condition on the walk not going from $j$ to $i$ in the first $n$ steps. 

For large enough $n$, $Pr[T^{\prime}_{k,i}>n]<Pr[T_{k,i}>n].$
So 
\[
E[T^{\prime}_{ki}]-E[T_{ki}]= \sum_{n=0}^{\infty} n(Pr[T^{\prime}_{k,i}>n]-Pr[T_{k,i}>n])<0.
\]

\section{Proof of Lemma 2}

The first part is easy: some new shortest paths form between previously disconnected vertices, and all of them pass through $i$. All other shortest paths are not affected. Hence betweenness centrality strictly increases by at least 1, unless $i$ was an isolated node (and hence there is no newly connected pair $s,t$. 

Consider now two nodes $s,t$ and let's add missing non-bridge edge $ij$. This has the following possible effects: \\
- \textbf{Case 1: The distance between  $s,t$ strictly decreases.} Then \textbf{all the (new) shortest paths} betweeen $s,t$ must pass through edge $ij$, i.e. through $i$. The fraction attributable to $s,t$ in the betweenness centrality of $i$ becomes 1, which is at least the corresponding fraction before adding edge $ij$. \\
- \textbf{Case 2: The distance between  $s,t$ stays the same.} Then all previous shortest paths remain shortest paths. On the other hand, by adding edge $ij$ one may add $c\geq 0$ new shortest paths between $s,t$ that use the edge $ij$. 
Since, for $a\leq b$ and $c\geq 0$, $\frac{a+c}{b+c}\geq \frac{a}{b}$, the fraction corresponding to $s,t$ in the formula for the betweenness centrality of $i$ weakly increases in this case as well. 

\section{Proof of Theorem 2}

\begin{proof} 

Let $H$ be an APSN. 

The formula for agents' utility is: 

\[
u_{i}(g)= \sum\limits_{j\in \widehat{N(i)}} \frac{1}{deg(j)+1} - c \cdot deg(i)
\]

Consider a pair $ij\not \in g$, and let $h=g+ij$. 

\begin{align*} 
& u_{i}(h)-u_{i}(g)= \frac{1}{deg(j)+2} +\frac{1}{deg(i)+2}-\frac{1}{deg(i)+1} = \\
& \frac{1}{deg(j)+2} -\frac{1}{(deg(i)+1)(deg(i)+2)}
\end{align*} 

Thus adding edge $ij$ to $g$ is an improving move for $i$ iff
\begin{align*} 
(deg(i)+1)(deg(i)+2)\geq deg(j)+3
\end{align*} 
i.e. iff $deg(j)\leq f(deg(i))$ with $$f(n)=(n+1)(n+2)-3$$. 
\end{proof} 

\section{Proof of Lemma 3}

Since $ij$ is a bridge edge, new nodes are reachable from $i$ by adding $ij$ (those of the connected component of $j$ in $g$), hence the maximum distance from $i$ cannot decrease.  

The second statement is equally simple: by adding an edge within a connected component shortest path distances cannot but decrease. So eccentricity (weakly) increases. 

\section{Proof of Corollary 1} 

We need the following simple 

\begin{lemma} 
Let $g$ be a network and $ij\in E(g)$. Then removing edge $ij$ from $g$ is an improving move  iff $deg(j)-1>  f(deg(i)-1)$ or $deg(i)-1> f(deg(j)-1)$. 
\label{aux}
\end{lemma} 
\begin{proof} 
Removing edge $ij$ is an improving move \textbf{iff} for at least one of the two nodes $i,j$ its betweenness centrality stays the same when removing the edge. In this case 
adding edge $ij$ to $h=g-ij$ is \textbf{not} an improving move and vice-versa: if adding edge $ij$ to $h$ is not an improving then one of $i,j$ has the same betweenness centrality in $g$ as in $h$, hence removing edge $ij$ is an improving move in $g$. 

By definition, adding edge $ij$ to $h$ is not improving iff 
$deg_{h}(i)>f(deg_{h}(j))$ or $deg_{h}(j)>f(deg_{h}(i))$. 
\end{proof} 

Consider now a sequence $a_{1}>a_{2}> \ldots > a_{p}> a_{p+1}=1$ be a sequence of integers such that $a_{i}-1 > f(a_{i+1}-1)$ for all $i=1,\ldots, p$. We first need to prove that all graphs of type 
$K_{a_{1}}+K_{a_{2}}+\ldots + K_{a_{p}}+lK_{0}$, $l\geq 0$, are APSN.  

This is easy, by applying points a),b),c) of the theorem: let, indeed, $y,z$ be nodes in the same clique $K_{a_r}$, $1\leq r \leq p$. We need to show that removing edge $yz$ is not an improving move. Since they are in the same clique, the degrees of $y,z$ are both equal to $a_r-1$. Since $a_r-1\leq f(a_r-1)$ 
(because $a_r\geq 2$ and $f(x)\geq x$ for $x\geq 1$), the desired conclusion follows by Lemma~\ref{aux}. 

Let now $y,z$ be nodes in different cliques, $y\in K_{a_r}$, $z\in K_{a_s}$, $a_r>a_s$. 
We have $deg(y)=a_r-1>f(a_s-1)=f(deg(z))$. By the definition, adding edge $yz$ is \textbf{not} an improving move. 

Since $0>f(0)=-1$ connecting any  isolated node to any other node is \textbf{not} an improving move. So $g$ is an APSN. 

Conversely, let $g$ be an APSN. By applying points a) and b) of the Theorem, we get that $g$ has edges between every two vertices whose degrees are in the same interval $[n_i^{*},n_{i-1}^{*}]$, where by convention $n_0^{*}=m$. 

To infer the fact that $g$ has the structure claimed in the corollary we need to prove that no other edges are present. Point c) of the theorem excludes edges between node whose degrees are \textbf{not} in the same interval. 

The only potential trouble is that there might be a node $x$ of degree $n_i^{*}$ who is connected with nodes whose degrees are in both intervals $[n_{i+1}^{*},n_{i}^{*}]$ and $[n_i^{*},n_{i-1}^{*}]$, thus "joining two cliques". We will show that something like this doesn't happen by induction on $i$. 

\textbf{Case $i=1$:} Let $z$ be a node of maximum degree $m$. Let $A$ be the set of nodes with degree in the range $[n_1^{*},n_{0}^{*}]$. Then all the nodes in $A$ are connected to each other. $z$ is not connected to any node outside $A$. If there were some other node $w$ in $A$ that is connected to a node outside $A$ then $w$ would have degree higher than $m$, a contradiction. Hence nodes in $A$ form a connected component that is a clique. 

\textbf{The induction step:} Assume we have obtained $l-1$ connected components that are cliques of size $a_1>a_2>\ldots > a_{l-1}$ satisfying the condition $a_{i}-1> f(a_{i+1}-1)$ for $i=1,\ldots, l-2$. Applying the reasoning in the induction case $i=1$ to the remaining graph we obtain a connected component of size $a_{l}$ that is a clique. Furthermore $a_{l-1}-1>f(a_{l}-1)$, since nodes in the $l$'th clique component are not connected to those in the $l-1$'st component. 

It is possible that the tail of the resulting sequence $a_1,\ldots, a_s$ is composed of components of size 1, that is isolated nodes. The required condition is satisfied, since $1-1>f(1-1)=f(0)=-1$. 

\section{Proof of Theorem 7}
Assume there was a node $p$ of degree 1. Let $q$ be the unique neighbor of $p$.  The degree of $q$ must be at least two, otherwise the connected component of $q$ would contain two vertices. Then $q$'s eccentricity would not change if removing the edge $pq$, so $q$ has an incentive to remove edge $pq$. 

The second part is easy as well: Consider an eccentricity-one APSN $g$. The eccentricity of the nodes in the center is one, i.e. center nodes are connected to all nodes. 

Suppose we remove some edge $ij$. As the minimum degree of the graph is two and its eccentricity is one, this removal does not disconnect the graph: if it did then $ij$ would be a bridge edge, and since $i,j$ have degree at least two, the eccentricity of the graph would be at least three. 

Thus the eccentricity of nodes $i,j$ goes from one to two, making their utility decrease for small enough costs. So the nodes $i,j$ do not have any incentive to remove their edge. 

Suppose there are two unconnected nodes $i,j$. 
Nor do any  two non-central nodes have any incentive to connect: there's still going to be another non-central node keeping their eccentricity to two. Thus the network is APSN. 

\section{Proof of Theorem 8}

\begin{proof} 
\begin{itemize} 
\item[a.] Let $\theta_{i}=C_i[g]$. We claim that $g$ is an APSN with respect to the truncated centrality game with centralities $C_i$ and thesholds $\theta_i$. 

Indeed, consider an edge $ij$ of $g$.  Nodes $i,j$ don't want to drop edge $ij$, since their current centrality values are $\theta_i,\theta_j$ while, by Axiom 1, their centralities would decrease below these values if they dropped $ij$. 

Let now $i,j$ be vertices such that $ij\not \in g$. Since their current centrality values are $\theta_i,\theta_j$ (at the threshold), adding edge $ij$ would not increase their truncated centralities, while incurring the extra (positive) cost of edge $ij$. So adding edge $ij$ is not an improving move. 

 \item[b.] First, by essentially repeating the proof at point a., it is easy to see that graphs with Pareto optimal centralities are APSN. 
 
 The opposite direction is equally easy: consider an APSN $h$ and two vertices $i,j$. If $ij\not \in E(h)$ then adding edge $ij$ must not be an improving move for one of $i,j$. Since $C_i,C_j$ are increasing, the only possibility is that $C_i [h]\geq \theta_i$ (so that adding edge $ij$ does not increase the truncated centrality of $i$ and, in fact, decrease its utility, because of the extra cost of edge $ij$) or, similarly,  $C_j[h]\geq \theta_j$. 
 
 Consider now the case when $ij\in h$. Because centralities are increasing and removing edge $ij$ is not an improving move, removing $ij$ must strictly decrease truncated centralities for both nodes $i,j$. This is only possible if the centralities of $i,j$ in the resulting network $l$ satisfy $C_i[l]<\theta_{i}$ and $C_j[l]<\theta_{j}$.

  \item[c.] First of all, a comment about the result of the previous section: it does \textbf{not} establish the existence of APSN, since it is not clear that the conditions in the characterization are actually feasible. This is what we show next: prove the existence of a network that satisfies them. 
  
 For games with truncated linear centralities we will prove the existence of APSN as follows: Sort the edges in decreasing weight order $w_{e_1}\geq w_{e_{2}}\geq \ldots \geq w_{e_{{n}\choose {2}}}$. Consider the following algorithm: 
 
 \begin{mdframed}
 \begin{itemize} 
 \item[] start with the empty graph $g=\emptyset_{n}$. 
 \item[] for $i=1$ to ${{n}\choose {2}}$:
 \item[] let $e_{i}=(a,b)$.
 \item[] if $C_{a}[g]<\theta_{a}$ and $C_{b}[g]<\theta_{b}$: 
 \item[] \hspace{5mm} $g=g\cup e_{i}$ 
 \item[] return $g$. 
 \end{itemize} 
 \end{mdframed}
 
 We claim that $g$ is an APSN. Indeed, none of the missing edges could be added because the condition is false: So if $ab\not \in E(g)$ then $C_{a}[g]\geq \theta_{a}$ or $C_{b}[g]\geq \theta_{b}$ at the moment when edge $at$ was considered for inclusion. Since centralities only increase during the algorithm, the condition is valid at the end of the algorithm as well. 
 
 On the other hand if $ab\in E(g)$ then at the moment node $a$ had its centrality $\geq \theta_{a}$ 
 no more edges adjacent to $a$ are added anymore. Since edges are in decreasing order, removing any edge adjacent to $a$ reduces the centrality of $a$ at least as much as the last added edge, thus bringing the centrality of $a$ below $\theta_{a}$. A similar argument works for node $b$.

 \item[d.] 
 To prove the existence of APSN for truncated centrality games satisfying the additional conditions in the hypothesis, consider a system of nonnegative thresholds $\theta_{i}$. Let $\mathcal{G}$ be the set of networks $g$ such that $C_l[g]\leq M_l$ for every vertex $l$. $\mathcal{G}$ is nonempty, since it contains the empty graph on $n$ vertices.

 Let $h$ be an edge-maximal member of $\mathcal{G}$. We claim that $h$ is an APSN. To prove this, will use the characterization of APSN as networks with  "Pareto optimal centralities". 
  
  The second condition in this characterization  is satisfied: indeed, assume that $ij\in E(h)$. Since $C_i[h]\leq  M_{i}$, $C_j[h]\leq 
 M_{j}$, by the hypothesis of the theorem it follows that $C_i[h-ij]< \theta_i$, and similarly for $j$. 
 
 As for the first condition, assuming by contradiction $ij\not \in E(h)$ but $C_i[h]<\theta_{i}$ and $C_j[h]<\theta_{j}$. By adding edge $ij$ (and obtaining network $l=h+ij$) by the condition in the hypothesis we would still satisfy condition $C_i[l]\leq 
M_{i}$ and $C_j[l]\leq 
 M_{j}$. Also, by Axiom~4, $C[k,l]\leq C[k,h]$, so 
 the condition $C[k,l]\leq M_{k}$ is still satisfied for other vertices $k\neq i,j$. This contradicts the maximality of $h$ in $\mathcal{G}$. So $ij\not \in E(h)$ implies that $C[j,h]\geq \theta_{i}$ and $C[j,h]\geq \theta_{j}$. 

 \begin{observation} In the conditions of point d., the reciprocal is also true: APSN are maximal edge-members of the set $\mathcal{G}$. 
 \end{observation}

\end{itemize} 
\end{proof}

\section{Proof of Theorem 9}

Assume that $\theta_{i}>C_{i}[\emptyset_{n}]$, otherwise the result is trivial. 

The algorithm is the obvious one: 

\begin{mdframed}
\begin{itemize} 
 \item[] start with the empty graph $g=\emptyset_{n}$. 
 \item[] while (oracle answer is not NONE):
 \item[]\hspace{5mm} query the oracle on pair $(g,i)$. 
 \item[]\hspace{5mm} let $h$ be the oracle answer
 \item[]\hspace{5mm} set $g=h$
 \item[] return $g$. 
 \end{itemize} 
 \end{mdframed}

Let $g$ be the APSN produced by the algorithm. Clearly, $\theta_{i}\leq C_{i}[g]$, since the algorithms returns an APSN with the highest centrality of $i$ which is, by the hypothesis, at least $\theta_{i}$. 

We claim that for every edge $ij\in E(g)$, $C_{i}[g-ij]<\theta_{i}$. Indeed, if this were not the case, then $i$ would have an incentive to drop edge $ij$, since its truncated centrality would not decrease, while its edge cost would. 
 
So every edge adjacent to $i$ is a good answer. 

In the case of linear centralities, the centrality of each node is integral, so any oracle step increases the centrality of $i$ by at least one. But the centrality of $i$ is upper bounded by the sum $\sum_{k} w_{i,k}$, so the algorithm runs in polynomial time in the problem size. 

\textbf{Note:} we assumed in this proof that integers $w_{i,k}$ are written in unary. Otherwise we would have a pseudo-polynomial/NP-complete complexity distinction similar to that of the  knapsack problem, depending on the representation.

\end{document}